\documentclass[conference]{IEEEtran}
\IEEEoverridecommandlockouts
\usepackage[english]{babel}
\usepackage[utf8]{inputenc}
\usepackage{ifpdf}
\usepackage{cite} 
\usepackage{url}
\usepackage[draft,bookmarks=false]{hyperref} 
\ifpdf
	\usepackage[pdftex]{graphicx}
	\graphicspath{{./figures/}}
\else
	\usepackage[dvips]{graphicx}
	\graphicspath{./figures/}
\fi
\usepackage{color}
\usepackage{pgf, tikz, pgfplots}
\usetikzlibrary{shapes, arrows, automata}
\usepackage{caption} 
\usepackage{subcaption} 
\usepackage{amsmath}
\usepackage{amsfonts, amssymb, amsthm}
\usepackage{mathrsfs}
\usepackage{upgreek}
\usepackage{algorithm,algpseudocode}
\usepackage{enumerate}
\usepackage{multirow}
\usepackage{rotating}
\usepackage{bm}
\usepackage{dsfont}
\usepackage{needspace}

\input{mySymbol.sty}

\newtheorem{assumption}{\hspace{0pt}\bf Assumption}

\newtheorem{proposition}{\hspace{0pt}\bf Proposition}

\newtheorem{theorem}{\hspace{0pt}\bf Theorem}

\newtheorem{definition}{\hspace{0pt}\bf Definition}

\newcommand{\QED}{\hfill\ensuremath{\blacksquare}}

\title{Stability of Neural Networks on \\ Riemannian Manifolds}

\author{Zhiyang~Wang, Luana~Ruiz~
        and~Alejandro~Ribeiro
\thanks{Supported by NSF CCF 1717120, Theorinet Simons. The authors are with the Dept. of Electrical and Systems Eng., Univ. of Pennsylvania. Email: \{zhiyangw,rubruiz,aribeiro\}@seas.upenn.edu.}
}

\begin{document}
%
\maketitle
%
\begin{abstract}
Convolutional Neural Networks (CNNs) have been applied to data with underlying non-Euclidean structures and have achieved impressive successes. This brings the stability analysis of CNNs on non-Euclidean domains into notice because CNNs have been proved stable on Euclidean domains. This paper focuses on the stability of CNNs on Riemannian manifolds. By taking the Laplace-Beltrami operators into consideration, we construct an $\alpha$-frequency difference threshold filter to help separate the spectrum of the operator with an infinite dimensionality. We further construct a manifold neural network architecture with these filters. We prove that both the manifold filters and neural networks are stable under absolute perturbations to the operators. The results also implicate a trade-off between the stability and discriminability of manifold neural networks. Finally we verify our conclusions with numerical experiments in a wireless adhoc network scenario.
\end{abstract}
%


\section{Introduction} \label{sec:intro}

Convolutional neural networks (CNNs) are machine learning architectures made up of layers where each layer composes a bank of convolutional filters with a pointwise nonlinearity. On problems where data is Euclidean, they have become a popular architecture due to their impressive performance in tasks ranging from speech recognition \cite{han2020contextnet} to computer vision \cite{gustafsson2020evaluating}, which is largely attributed to the fact that CNNs are provably stable \cite{bruna2013invariant}. 
But in the physical world, their application is limited because many problems deal with data that is non-Euclidean. This is the case, for instance, of resource allocation in wireless ad-hoc communication networks \cite{wang2020unsupervised}, detection and recognition in social networks \cite{nguyen2017robust} and prediction of influenza epidemic outbreaks \cite{wang2019defsi}. 


In recent years, a myriad of extensions of CNNs to non-Euclidean domains have been proposed to fill that gap \cite{bronstein2017geometric,gama2019convolutional, scarselli2008graph, defferrard2020deepsphere}.
These architectures have been able to reproduce the successes of CNNs on Euclidean domains to a large extent \cite{xu2018powerful, wu2020comprehensive}, which naturally sparks the question of whether non-Euclidean CNNs are also stable. In this paper, we aim to answer this question by analyzing the properties of CNNs defined on the most general type of non-Euclidean domain --- the manifold. By focusing on manifolds, our analysis has the benefit of being broad enough so that it can also be particularized to more specific non-Euclidean domains such as graphs, which can be seen as manifold discretizations. 

Explicitly, we study the stability of manifold CNNs to absolute perturbations of the Laplace-Beltrami operator $\ccalL$ associated with the manifold (Definition \ref{defn:absolute_perturbations}). We start by analyzing the stability of the convolution operation, which is defined as a pointwise operation on the spectrum of $\ccalL$. Given that absolute perturbations of $\ccalL$ spawn absolute perturbations to all of its eigenvalues, designing stable manifold convolutions is challenging because the spectrum of the Laplace-Beltrami operator is infinite-dimensional. We address this challenge by introducing frequency difference threshold (FDT) filters (Definition \ref{def:alpha-filter}), which separate the Laplacian spectrum into groups of eigenvalues that are less than some threshold $\alpha$ apart. We then show that these filters are stable to absolute perturbations of $\ccalL$ (Theorem \ref{thm:stability_filter}), and that this property is inherited by manifold CNNs (Theorem \ref{thm:stability_nn}). The main implication of these results is that there is a trade-off between the stability and discriminability of manifold neural networks in the form of the frequency difference threshold of the FDT filters.

Related work includes a comprehensive study of the stability of graph neural networks (GNNs) in \cite{gama2019stability} and \cite{gama2020stability}, which consider absolute and relative perturbations of the graph structure respectively, and the GNN stability analysis in \cite{zou2020graph}, which focuses on perturbations of the graph spectrum.
More in line with our paper, \cite{ruiz2020graph} studies stability of GNNs to perturbations of the underlying graph model, which is assumed to be a graphon. Unlike manifolds, however, graphons can only model dense graphs. 
More flexible models such as continuous graph models with tunable sparsity and generic topological spaces have been considered in \cite{keriven2020convergence} and \cite{levie2019transferability} respectively, but these papers focus on the transferability and not on the stability of convolutional neural networks in these domains.
The rest of this paper is organized as follows. We start with a brief review of manifolds and Laplacian operators in Section \ref{sec:prelim}. We further introduce the framework of manifold convolutions and neural networks. In Section \ref{sec:filters}, we introduce FDT filters and prove that the filters are stable under absolute perturbations of the Laplacian operator. We then extend this analysis to neural networks. Our results are verified numerically on a power allocation problem in wireless adhoc networks in Section \ref{sec:sims}, and conclusions are presented in Section \ref{sec:conclusions}.


\section{Preliminary definitions} \label{sec:prelim}

In order to analyze the stability properties of neural networks on manifolds, we start by reviewing the notions of signals, Laplacian operators and convolutional neural networks in these domains.

\subsection{Manifolds and manifold signals}
{A manifold is a topological space that is locally Euclidean around each point. More formally, a differentiable $d$-dimensional manifold is a topological space where each point has a neighborhood that is homeomorphic to a $d$-dimensional Euclidean space, i.e., the tangent space.} A Riemannian manifold, denoted $(\ccalM,g)$, is a real and smooth manifold $\ccalM$ equipped with a positive definite inner product $g_x$ on the tangent space $T_x\ccalM$ at each point $x$. The collection of tangent spaces of all points on $(\ccalM, g)$ is denoted as $T\ccalM$, and the collection  
of scalar functions and tangent vector functions on $(\ccalM,g)$ are denoted as $L^2(\ccalM)$ and $L^2(T\ccalM)$ respectively. In this paper, we consider compact Riemannian manifolds $\ccalM$.

\begin{figure}[h]
    \centering
    \includegraphics[width=0.4\textwidth]{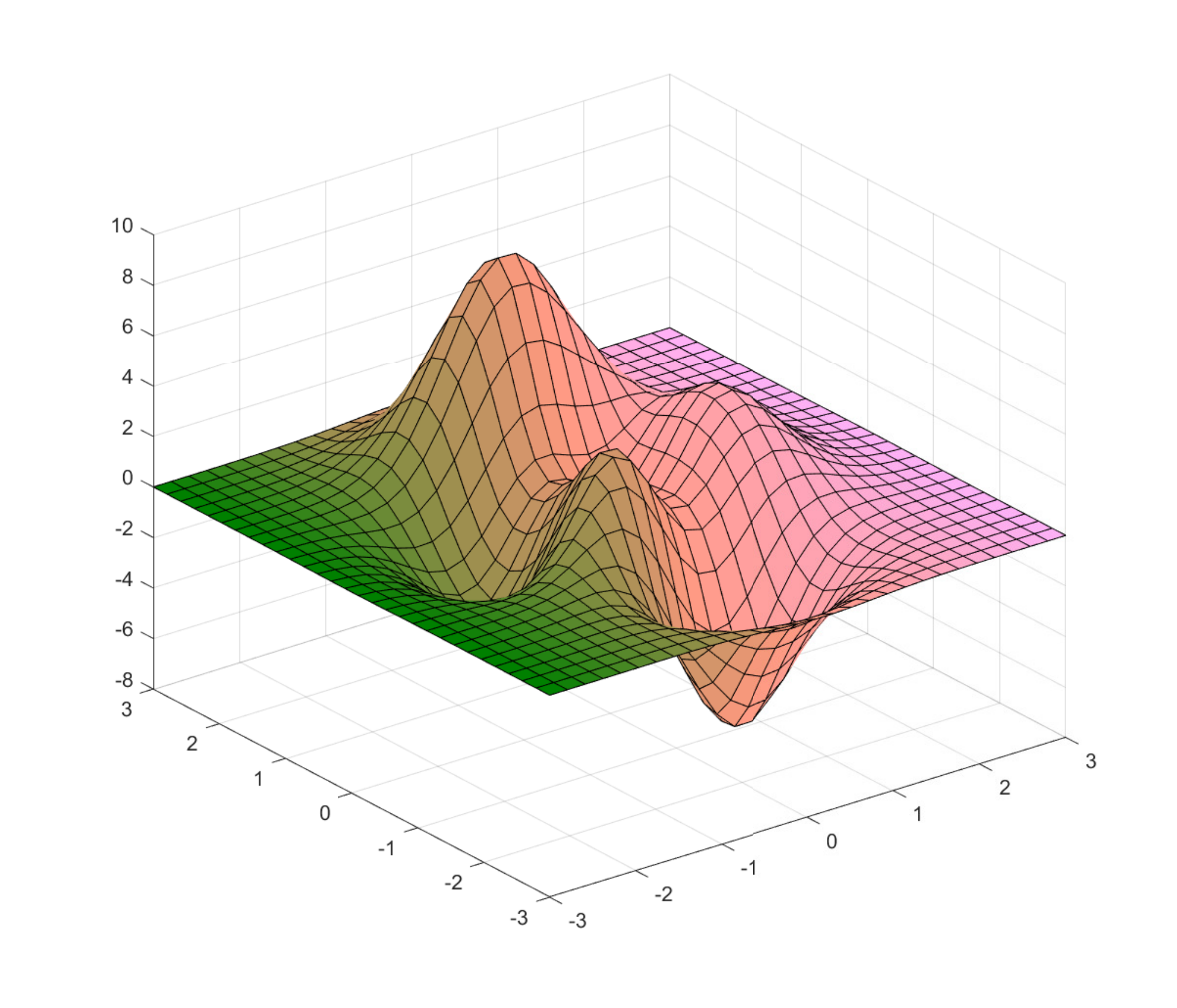}
    \caption{An example of manifold signal where the curve is the underlying manifold structure. Colors stand for values of the signal at the point.}
    \label{fig:manifold}
\end{figure}
Data supported on the manifold $\ccalM$ is represented as manifold signals, which are defined as functions $f\in L^2(\ccalM)$.
Given a manifold signal $f\in L^2(\ccalM):\ccalM\rightarrow \reals$, the Laplace-Beltrami operator $\ccalL$ is defined as
\begin{equation}\label{eqn:Laplacian}
    \ccalL f=-\text{div}(\nabla f),
\end{equation}
where $\nabla: L^2(\ccalM)\rightarrow L^2(T\ccalM)$ is an operator called \emph{intrinsic gradient}. Compared with the classical notion of gradient, which indicates the direction of the fastest change of a function at a given point, the difference here is that the direction indicated by the intrinsic gradient is a tangent vector in $L^2(T\ccalM)$. The operator $\text{div}$ 
is the \emph{intrinsic divergence}, which is adjoint to the gradient operator. Akin to the Laplacian operator in Euclidean domains or the graph Laplacian in graph signal processing (GSP), the Laplace-Beltrami operator defines a notion of shift for signals on a manifold $\ccalM$. In fact, $\ccalL$ has a similar interpretation to that of the graph Laplacian as the difference between the local average function value around a point and the function value at the point itself \cite{bronstein2017geometric}. 

The Laplacian operator $\ccalL$ is a self-adjoint and positive-semidefinite operator. As such, it admits an eigendecomposition given by
\begin{equation} \label{eqn:Laplacian-spectrum}
    \ccalL f=\sum_{i=1}^\infty \lambda_i\langle f, \bm\phi_i \rangle \bm\phi_i.
\end{equation}
where $\lambda_i$ are real eigenvalues and $\bm\phi_i$ are the corresponding eigenfunctions. Because the manifold $\ccalM$ is compact, the spectrum $\{\lambda_i,\bm\phi_i\}_{i\in\naturals^+}$ is additionally discrete. {If the eigenvalues are ordered in increasing order as $0<\lambda_1\leq \lambda_2\leq\lambda_3\leq \hdots$, it can be shown that, for a manifold on $d$ dimensions, $\lambda_i$ grows as $i^{2/d}$ \cite{arendt2009weyl}. 

Similarly to how the eigenvalues of the graph Laplacian are interpreted as frequencies in GSP, we interpret the Laplacian-Beltrami operator eigenvalues $\lambda_i$ as frequencies associated with oscillation modes $\bm\phi_i$.
Moreover, since the eigenfunctions $\bm\phi_1,\bm\phi_2,\hdots$ form an orthonormal basis of $L^2(\ccalM)$, a change of basis operation for signals $f$ can be defined, i.e., manifold signals $f\in L^2(\ccalM)$ can be represented on the manifold's eigenbasis as $f=\sum_{i=1}^\infty \langle f, \bm\phi_i \rangle \bm\phi_i$.}

\subsection{Manifold convolutions and manifold neural networks}

Drawing a parallel with the spectral convolution operation on Euclidean domains, we leverage the eigendecomposition of the Laplacian \eqref{eqn:Laplacian-spectrum} to define the convolution of manifold signals as a pointwise operation in the spectrum of $\ccalL$. Explicitly, we define a manifold convolutional filter \bbh(\ccalL) as
\begin{equation}\label{eqn:operator}
\bbh(\ccalL) f:=\sum_{i=1}^\infty \sum_{k=0}^{K-1} h_k \lambda_i^k \langle f,\bm\phi_i \rangle \bm\phi_i
\end{equation}
where $h_0, \ldots, h_{K-1}$ are the filter coefficients or taps.

Projecting $\bbh(\ccalL)f$ onto $\{\bm\phi_i\}_{i \in \naturals^+}$, we see that the spectral response of the manifold convolution is is given by $h(\lambda) = \sum_{k=0}^{K-1} h_k \lambda^k$. This highlights the fact that the frequency response of a manifold convolution only depends on the coefficients $h_k$ and the eigenvalues of the Laplace-Beltrami operator. It also implies that, if the underlying manifold changes --- and thus $\ccalL$ ---, the behavior of the filter $\bbh$ can be replicated on the new manifold $\ccalM'$ by evaluating $h(\lambda)$ at the eigenvalues of the new Laplacian $\ccalL'$. Another important consequence of the spectral representation of $\bbh(\ccalL)$ being a polynomial is that, as $K \to \infty$, $\bbh(\ccalL)$ can be used to implement any smooth spectral response $h(\lambda)$ with convergent Taylor series around each $\lambda_0 \in \reals$ as $K \to \infty$ \cite{smyth2014polynomial}.


From the definition of the convolution operation on $\ccalM$ \eqref{eqn:operator}, convolutional neural networks (CNNs) are straightforward to define. A CNN consists of a cascade of $L$ layers, each of which contains a bank of convolutional filters followed by a nonlinear activation function. Denoting the nonlinearity $\sigma$, the $l$-th layer of a $L$-layer CNN on the manifold $\ccalM$ is given by:
\begin{equation}\label{eqn:mnn}
f_l^p(x) = \sigma\left( \sum_{q=1}^{F_{l-1}} \bbh_l^{pq}(\ccalL) f_{l-1}^q(x)\right)
\end{equation}
for $l=1,2,\hdots,L$. Each of the filters $\bbh_l^{pq}(\ccalL)$ is as in \eqref{eqn:operator} and maps the $q$-th feature from the $l-1$-th layer to the $p$-th feature of the $l$-th layer for $1\leq q\leq F_{l-1}$ and $1\leq p\leq F_{l}$. The output of this neural network is given by $y^p = f_L^p$ for $1 \leq p \leq F_L$ and the input features at the first layer, $f_0^q$, are the input data $f^q$ for $1\leq q\leq F_0$. For a more concise representation of this CNN, we can alternatively write it as the map $\bby = \bbPhi(\bbH,\ccalL, f)$, where $\bbH$ is a tensor gathering the learnable parameters $\bbh_l^{pq}$ at all layers of the CNN. We will refer to this map as a manifold convolutional neural network or manifold neural network (MNN) for short.


\section{Stability of Manifold Neural Networks} \label{sec:filters}


{In order to characterize the stability properties of MNNs, we first have to study the stability of their main component --- the convolutional filter in \eqref{eqn:operator}. In particular, we analyze the stability of convolutional filters to absolute perturbations of the Laplace-Beltrami operator, which are specified in Definition \ref{defn:absolute_perturbations}.}


\begin{definition}[Absolute perturbations] \label{defn:absolute_perturbations}
Let $\ccalL$ be the Laplace-Beltrami operator of a Riemannian manifold $\ccalM$. An absolute perturbation of $\ccalL$ is defined as
\begin{equation}\label{eqn:perturb}
\ccalL'=\ccalL+\bbA,
\end{equation}
where the absolute perturbation operator $\bbA$ is symmetric.
\end{definition}

The absolute perturbation model introduced in Definition \ref{defn:absolute_perturbations} only requires $\bbA$ to be symmetric. Thus, it is a rather generic perturbation model encompassing many different types of perturbations and, in particular, allowing to model a wide array of perturbations to the underlying manifold $\ccalM$.

\subsection{Frequency difference threshold (FDT) filters}

{Given the spectral decomposition of the Laplace-Beltrami operator \eqref{eqn:Laplacian}, we can expect an absolute perturbation of $\ccalL$ to spawn some sort of perturbation to the eigenvalues $\lambda_i$. Since the spectral convolution operation \eqref{eqn:operator} depends on the evaluation of $h(\lambda)$ at each $\lambda_i$, its stability analysis will depend on the individual effects of the perturbation on each of these eigenvalues. A challenge in the case of manifolds is that the spectrum of $\ccalL$ is infinite-dimensional, i.e., there is an infinite (albeit countable) number of eigenvalues $\lambda_i$. However, it is possible to show that these eigenvalues accumulate in certain parts of the spectrum. This is demonstrated in Proposition \ref{prop:finite_num}.}


\begin{proposition} \label{prop:finite_num}
Let $(\ccalM,g)$ be a $d$-dimensional Riemannian manifold with Laplacian-Beltrami operator $\ccalL$, and let $\lambda_k$ denote the eigenvalues of $\ccalL$. Let $C_1$ denote an arbitrary constant and let $C_d$ be the volume of the $d$-dimensional unit ball. For any $\alpha > 0$, there exists $N_1$ given by
\begin{equation}
    N_1=\lceil (\alpha d/C_1)^{d/(2-d)}(C_d \text{Vol}(\ccalM,g))^{2/(2-d)} \rceil
\end{equation}
such that, for all $k>N_1$, it holds that $$\lambda_{k+1}-\lambda_k\leq \alpha.$$
\end{proposition}
\begin{proof}
This is a direct consequence of Weyl's law \cite{arendt2009weyl}.
\end{proof}
%

{Proposition \ref{prop:finite_num} is important because it suggests a strategy to mitigate the challenge posed by the infinite-dimensional spectrum of $\ccalL$. Since eigenvalues accumulate in certain parts of this spectrum, for $\alpha >0$ we can gather eigenvalues that are less than $\alpha$ apart in a \textit{finite} number of groups. This $\alpha$-separated spectrum, formalized in Definition \ref{def:alpha-spectrum}, is achieved by the so-called frequency difference threshold (FDT) filters introduced in Definition \ref{def:alpha-filter}.}







\begin{definition}[$\alpha$-separated spectrum]\label{def:alpha-spectrum}
The $\alpha$-separated spectrum of a Laplace-Beltrami operator $\ccalL$ is defined as the union of the set of $\alpha$-separated eigenvalues $\ccalD$,
$$\ccalD =\{\lambda_i : i\geq 1,\ \lambda_{i}-\lambda_{i-1}>\alpha,\ \lambda_{i+1}-\lambda_i>\alpha \}$$
and the set of $\alpha$-close eigenvalues $\ccalN = \bar{\ccalD}$,
\begin{align*}
\ccalN:=\{\ccalN_1,\ccalN_2,\cdots,\ccalN_N\} \mbox{ s.t. } \min|\lambda_i-\lambda_j|>\alpha \\
\mbox{ for all } \lambda_i\in\ccalN_a, \lambda_j\in\ccalN_b, a\neq b.
\end{align*}
\end{definition}

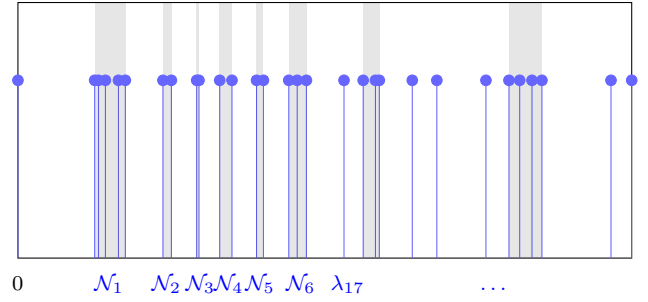
\begin{figure}[t]
    \centering

\pgfplotsset{xtick style={draw=none}}

\def \thisplotscale {3.4}
\def \unit {\thisplotscale cm}

\def \frequencyresponse 
     {   0.8}

\begin{tikzpicture}[x = 1*\unit, y=1*\unit]
 
\begin{axis}[scale only axis,
             width  = 2.4*\unit,
             height = 1*\unit,
             xmin = 0, xmax=8,
             xtick = { 0,  1.14 ,  1.89,2.33, 2.7 ,3.11, 3.64, 4.25, 6.1},
             xticklabels = {\black{\footnotesize $0$},
             				\blue{\footnotesize $\ \ccalN_1$}, 
                            \blue{\footnotesize $\ \ccalN_2$}, 
                            \blue{\footnotesize $\ \ccalN_3$}, 
                            \blue{\footnotesize $\ \ccalN_4$}, 
                            \blue{\footnotesize $\ \ccalN_5$}, 
                            \blue{\footnotesize $\ \ccalN_6$}, 
                            \blue{\footnotesize $\ \lam_{17}$},
                            \blue{\footnotesize $\ \ \ldots$},},
             ymin = -0, ymax = 1.15,
             ytick = {-1},
             typeset ticklabels with strut,
             enlarge x limits=false]

\addplot+[samples at = { 0, 1.00, 1.05,  1.14,
  1.31,  1.40, 1.89, 2.00,2.33,  2.36,    2.63, 2.79,      
   3.11 ,3.20,  3.53, 3.64, 3.76, 4.25, 4.50,4.66, 4.71, 5.14,  5.46, 6.10, 6.40,6.54, 6.70, 6.83, 7.73, 8}, 
          color = blue!60, 
          ycomb, 
          mark=otimes*, 
          mark options={blue!60}]
         {\frequencyresponse};
         
\addplot [fill=black, fill opacity=0.1, draw opacity = 0]
       coordinates {
            (1, 0) (1.4, 0) (1.4, 1.15) (1, 1.15)  };       
            
\addplot [fill=black, fill opacity=0.1, draw opacity = 0]
       coordinates {
            (1.89, 0) (2, 0) (2, 1.15) (1.89, 1.15)  };        
            
\addplot [fill=black, fill opacity=0.1, draw opacity = 0]
       coordinates {
            (2.33, 0) (2.36, 0) (2.36, 1.15) (2.33, 1.15)  };
            
\addplot [fill=black, fill opacity=0.1, draw opacity = 0]
       coordinates {
            (2.63, 0) (2.79, 0) (2.79, 1.15) (2.63, 1.15)  };
            
\addplot [fill=black, fill opacity=0.1, draw opacity = 0]
       coordinates {
            (3.11, 0) (3.2, 0) (3.2, 1.15) (3.11, 1.15)  };
            
\addplot [fill=black, fill opacity=0.1, draw opacity = 0]
       coordinates {
            (3.53, 0) (3.76, 0) (3.76, 1.15) (3.53, 1.15)  };
            
\addplot [fill=black, fill opacity=0.1, draw opacity = 0]
       coordinates {
            (4.5, 0) (4.71, 0) (4.71, 1.15) (4.5, 1.15)  }; 
            
\addplot [fill=black, fill opacity=0.1, draw opacity = 0]
       coordinates {
            (6.4, 0) (6.83, 0) (6.83, 1.15) (6.4, 1.15)  };            


\end{axis}
\end{tikzpicture}

    \caption{Eigenvalues of a Laplacian operator. Observe that the eigenvalues of the Laplacian operator tend to be grouped in certain parts of the spectrum.}
    \label{fig:eigenvalues_laplacian}
\end{figure}



\begin{definition}[$\alpha$-FDT filter]\label{def:alpha-filter}
On the manifold $\ccalM$, an $\alpha$-frequency difference threshold ($\alpha$-FDT) filter is a filter whose frequency response $h(\lambda)$ satisfies
\begin{equation} \label{eq:fdt-filter}
    h(\lambda)=C_n,\quad \lambda\in [\min_{i\in\ccalN_n}\lambda_i,\max_{j\in\ccalN_n}\lambda_j].
\end{equation}
for all $n=1,2,\hdots,N$. This makes the convolution operation as:
\begin{equation} \label{eqn:fdt-convolution}
\bbh(\ccalL)=\sum_{i\in \ccalD} h(\lambda_i)\langle f,\bm\phi_i\rangle \bm\phi_i +\sum_{n=1}^N C_n\langle f, E_n\rangle E_n,
\end{equation}
where $E_n$ is the eigenspace composed of $\{\bm\phi_i\}_{i\in\ccalN_n}$.
\end{definition}

{The eigenvalues $\lambda_i$ in the set $\ccalD$ are at least $\alpha$ apart from their preceding and succeeding eigenvalues $\lambda_{i-1}$ and $\lambda_{i+1}$. Hence, according to Definition \ref{def:alpha-filter}, an FDT filter with threshold $\alpha$ treats them as independent eigenvalues, like a conventional convolutional filter \eqref{eqn:filter_function} would.
The eigenvalues in the complement set $\ccalN = \bar\ccalD$ are divided in groups $\ccalN_a$ where each group contains eigenvalues that are less than $\alpha$ apart from at least one their neighbors. As such, the distance between any two groups $\ccalN_a, \ccalN_b$, $a \neq b$, is larger than $\alpha$. To achieve spectrum separation, the $\alpha$-FDT filter imposes a constant frequency response $h(\lambda_i)=C_a$ for all $i \in \ccalN_a$.} In other words, it treats all the eigenvalues in a group $\ccalN_a$ as the same. Note that, while the $h(\lambda)$ in Definition \ref{def:alpha-filter} is not a smooth function, we can still obtain a smooth approximation of an FDT filter in the form of \eqref{eqn:operator}.


\begin{figure}[t]
    \centering
    \input{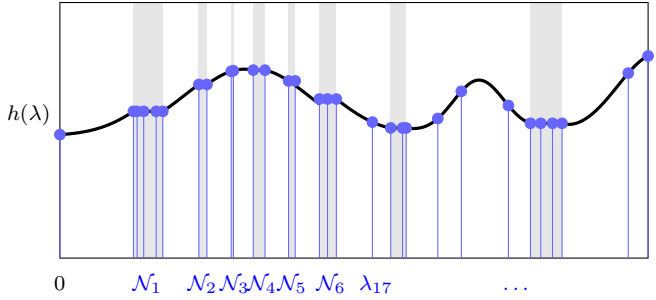}    
    \caption{An $\alpha$-FDT filter that separates the spectrum of the Laplacian operator by grouping eigenvalues that are less than $\alpha$ apart in sets $\ccalN_i$. Note that the frequency response is constant for all eigenvalues in a group $\ccalN_i$.}
    \label{fig:filter-response}
\end{figure}

\subsection{FDT filter stability}

{Thanks to spectrum separation, $\alpha$-FDT manifold filters can be shown to be stable under the absolute perturbations to the Laplace-Beltrami operator in Definition \ref{defn:absolute_perturbations}. 
We can state and prove this in Theorem \ref{thm:stability_filter} under Assumption \ref{ass:filter_function}.}

\begin{assumption}\label{ass:filter_function}
The filter function $h:\reals\rightarrow\reals$ is $B$- Lipschitz continuous and non-amplifying, i.e.,
\begin{equation}\label{eqn:filter_function}
    |h(a)-h(b)|\leq B|a-b|,\quad |h(a)|< 1.
\end{equation}
\end{assumption}

\begin{theorem}[FDT filter stability]\label{thm:stability_filter}
Let $\ccalM$ be a manifold with Laplace-Beltrami operator $\ccalL$. Let $f$ be a manifold signal and $\bbh(\ccalL)$ an $\alpha$-FDT filter on $\ccalM$ [cf. Definition \ref{def:alpha-filter}]. Consider an absolute perturbation $\ccalL'=\ccalL + \bbA$ of the Laplace-Beltrami operator $\ccalL$ [cf. Definition \ref{defn:absolute_perturbations}] where $\|\bbA\| = \epsilon \leq \alpha$. Then, under Assumption \ref{ass:filter_function} it holds
\begin{equation}\label{eqn:stability_filter}
\| \bbh(\ccalL) f -\bbh(\ccalL')f \| \leq  \frac{\pi(D+N)}{\alpha-\epsilon}\epsilon\|f\| + BD\epsilon\|f\|
\end{equation}
where $D$ and $N$ are the cardinalities of sets $\ccalD$ and $\ccalN$.
\end{theorem}

Provided that Assumption \ref{ass:filter_function} is satisfied, FDT filters are thus stable to absolute perturbations of the operator $\ccalL$. Note that this assumption is rather reasonable, as there are no constraints on the value of the Lipschitz constant $B$. The stability bound depends on (i) the variability of the FDT filter as measured by the Lipschitz constant $B$ and (ii) its frequency difference threshold, both directly through $\alpha$, and indirectly through $D$ and $N$. It is also split into two terms. The first one arises from the difference between the eigenfunctions of $\ccalL$ and $\ccalL'$, 
while the second is a result of the distance between their eigenvalues.

We observe that stability is improved if the Lipschitz constant $B$ is small. However, this causes the filter to become less discriminative and give similar response to all spectral components. With a larger $\alpha$, the FDT filter identifies fewer eigenvalues as being $\alpha$-separated, which decreases $D$. While this reflects an increase in the number of groups $N$, a certain number of eigenvalues previously in $\ccalD$ ends up being replaced by a fewer number of groups in $\ccalN$, leading $D+N$ to also decrease. Thus, larger values of $\alpha$ improve stability, but this also happens at the cost of discriminability as filters with large $\alpha$ separate the spectrum more sparsely, i.e., they identify fewer eigenvalues as being $\alpha$-separated. Interestingly, this stability-discriminability trade-off does not depend on the magnitude of the frequencies amplified by the filter (as is the case in, e.g., \cite{gama2020stability, ruiz2020graphon}). Instead, it is associated with the filter's precision in telling neighboring frequencies apart wherever they mare in the spectrum.

\subsection{Neural network stability} \label{sec:mnns}
Manifold neural networks with banks of filters like the one in \eqref{eq:fdt-filter} inherit the stability properties of $\alpha$-FDT filters. This is demonstrated in Theorem \ref{thm:stability_nn} under Assumption \ref{ass:activation}.

\begin{assumption}\label{ass:activation}
 The activation function $\sigma$ is normalized Lipschitz continous, i.e., $|\sigma(a)-\sigma(b)|\leq |a-b|$, with $\sigma(0)=0$.
\end{assumption}
 
\begin{theorem}[Neural network stability]\label{thm:stability_nn}
 Let $\ccalM$ be a manifold with Laplace-Beltrami operator $\ccalL$. Let $f$ be a manifold signal and $\bm\Phi(\bbH,\ccalL,f)$ an $L$-layer manifold neural network on $\ccalM$ \eqref{eqn:mnn} with $F_0=F_L=1$ input and output features and $F_l=F,i=1,2,\hdots,L-1$ features per layer, and where the filters $\bbh(\ccalL)$ are $\alpha$-FDT filters [cf. Definition \ref{def:alpha-filter}]. 
 Consider an absolute perturbation $\ccalL'=\ccalL + \bbA$ of the Laplace-Beltrami operator $\ccalL$ [cf. Definition \ref{defn:absolute_perturbations}] where $\|\bbA\| = \epsilon \leq \alpha$. 
 Then, under Assumptions \ref{ass:filter_function} and \ref{ass:activation} it holds:
 \begin{align}\label{eqn:stability_nn}
 \begin{split}
    \|\bm\Phi(\bbH,\ccalL,f)-&\bm\Phi(\bbH,\ccalL',f)\| \\
    &\leq LF^{L-1}\left(\frac{\pi(D+N)}{\alpha-\epsilon}+BD\right)\epsilon \|f\|.
\end{split}
 \end{align}
where $D$ and $N$ are the cardinalities of sets $\ccalD$ and $\ccalN$.
 \end{theorem}
 
Deep neural networks on manifolds are thus also stable to absolute perturbations provided that the activation function is normalized Lipschitz. This assumption is satisfied by most common activation functions, such as the ReLU, the modulus function and the sigmoid. Here, the same general comments as in the case of Theorem \ref{thm:stability_filter} hold, with the difference that in Theorem \ref{thm:stability_nn} the stability bound also depends on the number of layers $L$ and the number of features per layer $F$ of the MNN.


{As graphs can be seen as discretizations of manifolds, therefore we can use graph neural networks to realize the function of manifold neural networks. Combined with the transferability analysis from manifolds to graphs, we could state that the stability result that we have in Theorem \ref{thm:stability_nn} can be extended to graph neural networks.}


\section{Numerical experiments} \label{sec:sims}

We take a graph neural network model to approximate the manifold neural network. We verify our results on a {graph neural network} supported on a wireless adhoc network with $n=50$ nodes within a range of $[-50m,50m]^2$ where nodes are placed randomly. The channel states of all links can be represented by a matrix $\bbH(t)$ with each element $[\bbH(t)]_{ij}:=h_{ij}(t)$ denotes the channel condition between node $i$ and node $j$. Consider the large-scale pathloss gain and a random fast fading gain, this can be written as: $h_{ij}=\log( d_{ij}^{-2.2} h^f)$, where $d_{ij}$ stands for the distance between node $i$ and $j$, while $h^f\sim  \text{Rayleigh}(2)$ is the random fading. We here study the power allocation problem among $n$ nodes over an AWGN channel with interference, with $\bbp(\bbH)=[p_1(\bbH),p_2(\bbH),\hdots,p_n(\bbH)]$ denoting the power allocated to each node under channel condition $\bbH$. the channel rate of node $i$ is represented as $r_i$. The goal is to maximize the sum rate capacity under a total power budget $P_{max}$. This can be numerically formulated as:
\begin{align}
\label{eqn:prob_sim}
r^*&=\max_{\bbp(\bbH)} \sum_{i=1}^n r_i\\
   s.t.\quad \nonumber &r_i=\mathbb{E}\left[  \log\left(1+\frac{|h_{ii}|^2 p_i(\bbH)}{1+ \sum\limits_{j\neq i} |h_{ij}|^2 p_j(\bbH)}\right) \right],\\
   \nonumber & \mathbb{E}[\bm{1}^T\bbp]\leq P_{max},\quad p_i(\bbH)\in \{0,p_0\}.
\end{align}

With the nodes and the links seen as the graph nodes and edges respectively, the channel matrix $\bbH$ can be seen as a graph shift operator, more specifically, an adjacency matrix. By formulating $\bbH$ into a Laplacian matrix, the problem can be solved with a graph neural network composed with our defined $\alpha$-FDT filters with $\alpha$ set as 0.001. By setting $h_k=h_j$ for $|\lambda_k-\lambda_j|\leq\alpha$ in \eqref{eqn:operator}, we can get an approximation of $\alpha$-FDT filters. After trained for $4000$ iterations, the graph neural network can achieve the optimal power allocation.  In physical world, the nodes are deployed in some specific spatial positions. The positions may change and this would cause perturbations to the underlying Laplacian matrix. To model this, we add a log-normal matrix to the original channel matrix $\bbH$. With the same trained graph neural network employed, we measure the stability by the difference of final sum-rate. 

\begin{figure}[t]
    \centering
    \includegraphics[width=0.4\textwidth]{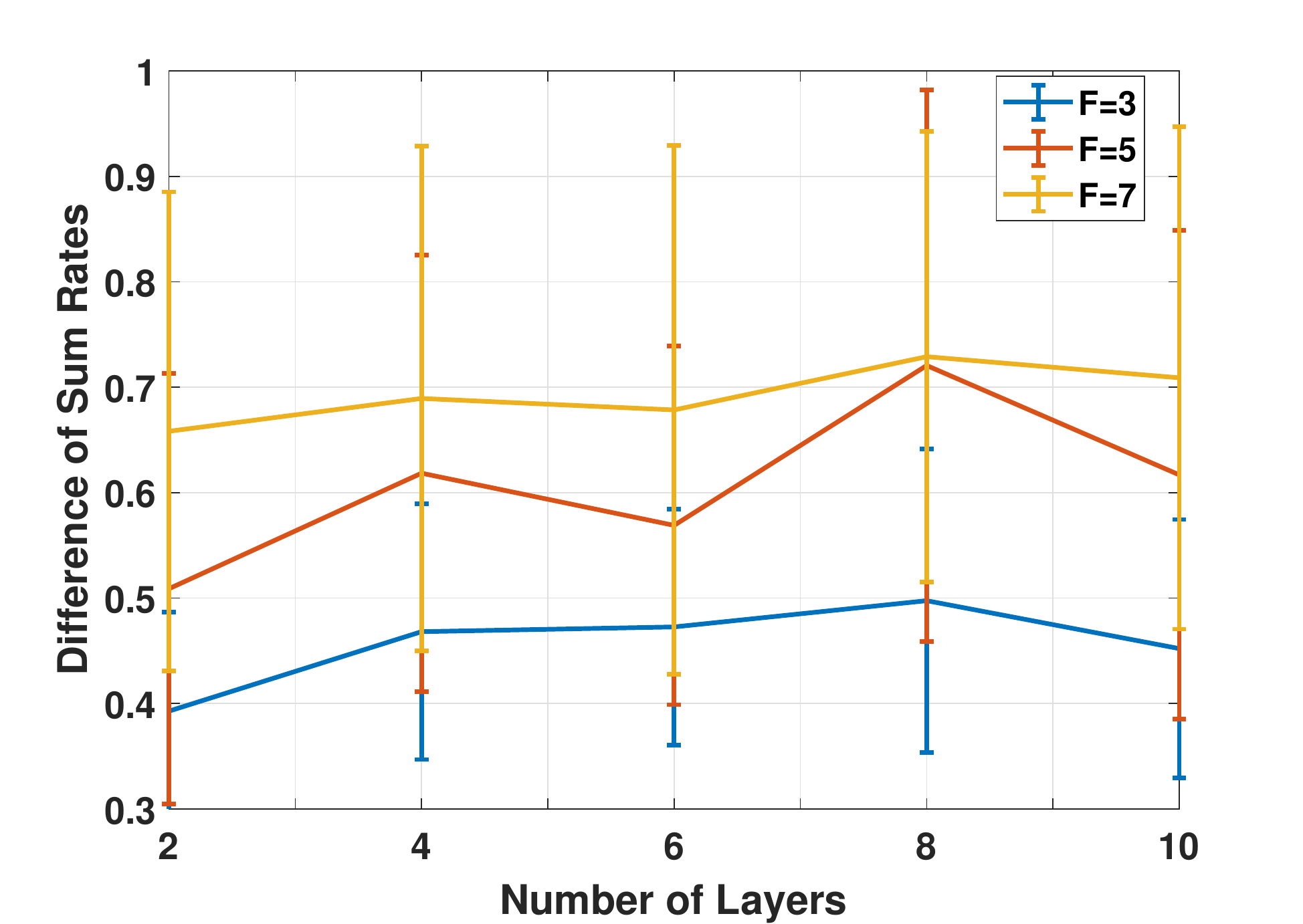}
    \caption{Sum-rate difference on the test between the original wirelss network setting and the perturbed one.}
    \label{fig:sim_stability}
\end{figure}

We can observe from Figure \ref{fig:sim_stability} that the difference of the final sum-rates increases with the number of layers and the number of filters per layer in the neural network. This verifies the conclusion that we have derived in Theorem \ref{thm:stability_nn}.


\section{Conclusions} \label{sec:conclusions}


In this paper, we have defined manifold convolutions and manifold neural networks. Considering the infinite dimensionality of Laplace-Beltrami operators. We import the definition of $\alpha$ frequency difference threshold filters to help separate the spectrum. By assigning a constant frequency response to the eigenvalues that are close enough, $\alpha$-FDT filters can be proved to be stable to absolute perturbations to the Laplacian operators. We further prove that the manifold neural networks built with $\alpha$-FDT filters are also stable under absolute perturbations. We conclude that there is a trade-off between the stability and discriminability. We further verified our results numerically with a power allocation problem in wireless adhoc networks.


\bibliographystyle{IEEEbib}
\bibliography{references}

\begin{thebibliography}{10}

\bibitem{han2020contextnet}
Wei Han, Zhengdong Zhang, Yu~Zhang, Jiahui Yu, Chung-Cheng Chiu, James Qin,
  Anmol Gulati, Ruoming Pang, and Yonghui Wu,
\newblock ``Contextnet: Improving convolutional neural networks for automatic
  speech recognition with global context,''
\newblock {\em arXiv preprint arXiv:2005.03191}, 2020.

\bibitem{gustafsson2020evaluating}
Fredrik~K Gustafsson, Martin Danelljan, and Thomas~B Schon,
\newblock ``Evaluating scalable bayesian deep learning methods for robust
  computer vision,''
\newblock in {\em Proceedings of the IEEE/CVF Conference on Computer Vision and
  Pattern Recognition Workshops}, 2020, pp. 318--319.

\bibitem{bruna2013invariant}
Joan Bruna and St{\'e}phane Mallat,
\newblock ``Invariant scattering convolution networks,''
\newblock {\em IEEE transactions on pattern analysis and machine intelligence},
  vol. 35, no. 8, pp. 1872--1886, 2013.

\bibitem{wang2020unsupervised}
Zhiyang Wang, Mark Eisen, and Alejandro Ribeiro,
\newblock ``Unsupervised learning for asynchronous resource allocation in
  ad-hoc wireless networks,''
\newblock {\em arXiv preprint arXiv:2011.02644}, 2020.

\bibitem{nguyen2017robust}
Dat Nguyen, Kamela~Ali Al~Mannai, Shafiq Joty, Hassan Sajjad, Muhammad Imran,
  and Prasenjit Mitra,
\newblock ``Robust classification of crisis-related data on social networks
  using convolutional neural networks,''
\newblock in {\em Proceedings of the International AAAI Conference on Web and
  Social Media}, 2017, vol.~11.

\bibitem{wang2019defsi}
Lijing Wang, Jiangzhuo Chen, and Madhav Marathe,
\newblock ``Defsi: Deep learning based epidemic forecasting with synthetic
  information,''
\newblock in {\em Proceedings of the AAAI Conference on Artificial
  Intelligence}, 2019, vol.~33, pp. 9607--9612.

\bibitem{bronstein2017geometric}
Michael~M Bronstein, Joan Bruna, Yann LeCun, Arthur Szlam, and Pierre
  Vandergheynst,
\newblock ``Geometric deep learning: going beyond euclidean data,''
\newblock {\em IEEE Signal Processing Magazine}, vol. 34, no. 4, pp. 18--42,
  2017.

\bibitem{gama2019convolutional}
Fernando Gama, Antonio~G Marques, Geert Leus, and Alejandro Ribeiro,
\newblock ``Convolutional neural network architectures for signals supported on
  graphs,''
\newblock {\em IEEE Transactions on Signal Processing}, vol. 67, no. 4, pp.
  1034--1049, 2019.

\bibitem{scarselli2008graph}
Franco Scarselli, Marco Gori, Ah~Chung Tsoi, Markus Hagenbuchner, and Gabriele
  Monfardini,
\newblock ``The graph neural network model,''
\newblock {\em IEEE transactions on neural networks}, vol. 20, no. 1, pp.
  61--80, 2008.

\bibitem{defferrard2020deepsphere}
Micha{\"e}l Defferrard, Martino Milani, Fr{\'e}d{\'e}rick Gusset, and
  Nathana{\"e}l Perraudin,
\newblock ``Deepsphere: a graph-based spherical cnn,''
\newblock {\em arXiv preprint arXiv:2012.15000}, 2020.

\bibitem{xu2018powerful}
Keyulu Xu, Weihua Hu, Jure Leskovec, and Stefanie Jegelka,
\newblock ``How powerful are graph neural networks?,''
\newblock {\em arXiv preprint arXiv:1810.00826}, 2018.

\bibitem{wu2020comprehensive}
Zonghan Wu, Shirui Pan, Fengwen Chen, Guodong Long, Chengqi Zhang, and S~Yu
  Philip,
\newblock ``A comprehensive survey on graph neural networks,''
\newblock {\em IEEE transactions on neural networks and learning systems},
  2020.

\bibitem{gama2019stability}
Fernando Gama, Joan Bruna, and Alejandro Ribeiro,
\newblock ``Stability of graph scattering transforms,''
\newblock {\em arXiv preprint arXiv:1906.04784}, 2019.

\bibitem{gama2020stability}
Fernando Gama, Joan Bruna, and Alejandro Ribeiro,
\newblock ``Stability properties of graph neural networks,''
\newblock {\em IEEE Transactions on Signal Processing}, vol. 68, pp.
  5680--5695, 2020.

\bibitem{zou2020graph}
Dongmian Zou and Gilad Lerman,
\newblock ``Graph convolutional neural networks via scattering,''
\newblock {\em Applied and Computational Harmonic Analysis}, vol. 49, no. 3,
  pp. 1046--1074, 2020.

\bibitem{ruiz2020graph}
Luana Ruiz, Zhiyang Wang, and Alejandro Ribeiro,
\newblock ``Graph and graphon neural network stability,''
\newblock {\em arXiv preprint arXiv:2010.12529}, 2020.

\bibitem{keriven2020convergence}
Nicolas Keriven, Alberto Bietti, and Samuel Vaiter,
\newblock ``Convergence and stability of graph convolutional networks on large
  random graphs,''
\newblock {\em arXiv preprint arXiv:2006.01868}, 2020.

\bibitem{levie2019transferability}
Ron Levie, Michael~M Bronstein, and Gitta Kutyniok,
\newblock ``Transferability of spectral graph convolutional neural networks,''
\newblock {\em arXiv preprint arXiv:1907.12972}, 2019.

\bibitem{arendt2009weyl}
Wolfgang Arendt, Robin Nittka, Wolfgang Peter, and Frank Steiner,
\newblock ``Weyl’s law: Spectral properties of the laplacian in mathematics
  and physics,''
\newblock {\em Mathematical analysis of evolution, information, and
  complexity}, pp. 1--71, 2009.

\bibitem{smyth2014polynomial}
Gordon~K Smyth,
\newblock ``Polynomial approximation,''
\newblock {\em Wiley StatsRef: Statistics Reference Online}, 2014.

\bibitem{ruiz2020graphon}
Luana Ruiz, Luiz~FO Chamon, and Alejandro Ribeiro,
\newblock ``Graphon neural networks and the transferability of graph neural
  networks,''
\newblock {\em arXiv preprint arXiv:2006.03548}, 2020.

\end{thebibliography}

\end{document}